\documentclass[conference]{IEEEtran}
\usepackage{amsfonts}
\usepackage{amssymb}
\usepackage{stfloats}
\usepackage{cite}
\usepackage{graphicx}
\usepackage{psfrag}
\usepackage{amsmath}
\usepackage{array}
\usepackage{epstopdf}
\usepackage{authblk}
\usepackage{graphicx} 
\usepackage{amsthm} 
\usepackage{lipsum}
\usepackage{verbatim} 
\usepackage{color}
\usepackage{caption}
\usepackage{subcaption}

\newtheorem{Lemma}{Lemma}
\newtheorem{Theorem}{Theorem}
\newtheorem{Algorithm}{Algorithm}

\newtheorem{Scheme}{Scheme}

\title{Learning-based Rate Adaptation for Uplink Massive MIMO Networks with Cooperative Data-Assisted Detection}

\author{Yang Li, Zezhong Zhang, Rui Wang, Kaibin Huang and Yifan Chen}

\begin{document}

\maketitle

\begin{abstract}
In this paper, the uplink adaptation for massive multiple-input-multiple-output (MIMO) networks without the knowledge of user density is considered. Specifically, a novel cooperative uplink transmission and detection scheme is first proposed for massive MIMO networks, where each uplink frame is divided into a number of data blocks with independent coding schemes and the following blocks are decoded based on previously detected data blocks in both service and neighboring cells. The asymptotic signal-to-interference-plus-noise ratio (SINR) of the proposed scheme is then derived, and the distribution of interference power considering the randomness of the users' locations is proved to be Gaussian. By tracking the mean and variance of interference power, an online robust rate adaptation algorithm ensuring a target packet outage probability is proposed for the scenario where the interfering channel and the user density are unknown.
%Moreover, the gradient of the average network goodput, measuring the average number of information bits per frame successfully delivered to the BSs in the considered network area, with respect to the data block lengths is derived, which depends on the unknown user density. A learning algorithm based on the principle of stochastic gradient descent is then proposed to adjust the data block lengths such that the average network goodput can be maximized.
\end{abstract}

\section{Introduction} \label{sec:intro}

As a technique offering sufficient degree of freedom and considerable spatial multiplexing gain by using large numbers of antennas, massive MIMO has a great potential to boost the throughput of cellular networks. However, the issue of {\em pilot contamination}\cite{Jose:09}, which refers to the undiminished inter-cell interference caused by pilot reuse, may severely degrade the performance of massive MIMO networks. One straightforward approach to mitigate the inter-cell interference in channel estimation is the {\em soft pilot reuse} scheme \cite{XudongZhu:2015,2017WCLZZZ}, where  pilot sequences are only reused at the center of cells and hence its overhead is large. 
%In \cite{Gesbert:12}, a joint channel estimation and pilot assignment mechanism was proposed to achieve interference free performance under a non-overlapping angle-of-arrival (AOA) multipath channel model. 
There are also some works introducing compressive-sensing-based channel estimation algorithms to exploit the sparsity in massive MIMO channel \cite{ Masood:2015}. All these works exploit the channel spatial correlation. In contrast, there are some works considering massive MIMO transmission schemes without such correlation.  For example in \cite{Fernandes:13}, the authors proposed to synchronize the uplink channel estimation with the downlink transmission of neighboring cells. However, this might cause severe interference between BSs in practice. Another promising approach addressing the issue of pilot contamination is to utilize the data symbols in channel estimation. Channel estimation schemes based on superimposed pilots was proposed in \cite{Upadhya:2017,2018Access}, where pilot and data symbols are transmitted alongside each other for the entire duration of the uplink frame. We proposed to utilize temporal channel correlation and detected data symbols to improve channel estimation in \cite{Liyang17} without the consideration of BS cooperation. Moreover, in \cite{2014Globecom,Ray:15,2015DSP,2017CL}, we proposed to divide the uplink frame into blocks and to utilize the detected data blocks as equivalent pilots for the successive data detection. In \cite{2017WCSP}, pilots of different lengths are assigned to cell-center users and cell-edge users to suppressing the pilot contamination efficiently.
%In most of the above works, the link-level performance metrics, including the throughput and SINR, are usually used to evaluate the proposed scheme. However, the performance evaluation from the network point of view is equally important if not more. For example, one transmission scheme may benefit the cell-edge users at the price of cell-center users, hence it is better to consider its performance at the network level. 

To study the network performance, it is important to consider the randomness on the users' geometric distribution. There exist a number of works on the massive MIMO performance analysis with the tool of stochastic geometry.  For example, an analysis on the asymptotic SINR distribution in the uplink was given in~\cite{StochasticGeo} with maximum ratio combining (MRC) and zero-forcing (ZF) receivers. 
%The spectrum efficiency and energy efficiency of $K$-tier heterogeneous networks (HetNets) were analyzed in~\cite{AnqiHe} with massive MIMO deployed in the macro cells, where the authors proposed a novel cell-association scheme to improve the energy efficiency of HetNets by offloading data traffic to small cells.
%The downlink spectral efficiency of the multicell massive MIMO was analyzed in~\cite{StochasticChau}, where the authors proposed a beamforming training (BT) and pilot-contamination-precoding (PCP) transmission scheme to mitigate pilot contamination with limited cooperation between BSs. 
However, in these works, the network performance is derived by randomizing the locations of both BSs and users. The conclusions of these works cannot be applied directly on one particular network with fixed locations of BSs. Nevertheless, although performance with fixed BSs' locations is studied in~\cite{2017GlobecomZZZ,AccessZZ,2018Globecom,2019TWC}, it assumes uniform distribution of users, which may not match the practical scenario. Moreover, the performance distribution derived in these works is complicated, and the further parameter adaptation based on it may be intractable.

\subsection{Our Contributions}

In this paper, we consider the uplink transmission of a massive MIMO network, where the user density is unknown. The main contributions are summarized below.
\begin{itemize}
	\item In this paper, we propose a scheme where the uplink frame is divided into data blocks with independent channel coding schemes, and the data symbols in the detected data blocks can be shared among neighboring BSs for the estimation of dominant interfering channels. Therefore, inter-cell interference from these channels can be mitigated, and better performance than the schemes in \cite{Ray:15,2015DSP} can be obtained. Moreover, we prove that the asymptotic distribution of interference power is Gaussian, even without uniform user distribution. 
	
	\item We consider the practical scenario that the user density is unknown, and the interfering channel is not measured at the service BSs. In this case, neither the accurate value nor the distribution of uplink SINR is known at the beginning of transmission, which leads to potential packet outage (the transmission data rate is larger than the channel capacity). Exploiting the fact that the interference is Gaussian, a robust rate adaptation algorithm given a target packet outage probability is proposed. Moreover, we also show that the distribution statistics of interference power can be analytically derived for all other users as long as they have been learned for a few number of users.	
%	\item Given the above robust rate adaptation scheme, the network performance depends on the choice of each block length and the user density. The latter is unknown. In order to optimize the block lengths, we first derive the gradient of the average network goodput, which measures the average number of information bits per frame successfully delivered to the BSs in the considered network region. Lack of the knowledge on the user density, the gradient cannot be directly calculated. Since its unbiased estimation can be obtained in each frame, we propose a stochastic gradient descent (SGD) algorithm for online block-length adaptation.
\end{itemize}

%The remaining of this paper is organized as follows. In Section \ref{sec:model}, the system model and the framework of stochastic-geometry-based performance analysis are introduced. In Section \ref{sec:scheme}, the data-assisted schemes is introduced respectively with analytical performance evaluation. In Section \ref{sec:learning}, the stochastic approximation based transmission adaptation is elaborated. In Section \ref{sec:sim}, the numerical simulation results are elaborated and compared with the analytical results. Finally, the conclusion is drawn in Section \ref{sec:sum}.

\section{System Model}\label{sec:model}

\subsection{Uplink Model of Massive MIMO Network}
The uplink transmission of a massive MIMO cellular network is considered, where BSs are deployed on a 2-dimensional plane $ \mathcal{R}^2 $. The number of antennas at each BS is $ M $. It is assumed that there is one BS locating at the origin. It is treated as the target BS of analysis for convenience. Each BS (or cell) is assigned with an index, and the index of the target BS (or cell) is $ 1 $. The location of the $ i $-th BS ($ \forall i $) is denoted by a 2-dimensional vector $ \mathbf l^i $. There can be multiple users, each with single antenna, in each cell. Due to random packet arrival, not every user has uplink data in each frame. We refer to the users with uplink data as active users. The distribution of all active uplink users in one frame, is modeled as an SPPP $ \Pi_{u} $ on $ \mathcal{R}^2 $ with density $ \lambda_{u}(\mathbf{l}) $ for location $ \mathbf{l} \in \mathcal{R}^2 $. The user density may vary for different locations. It is assumed that $ 0 < \lambda_u^l \leq \lambda_{u}(\mathbf{l}) \leq \lambda_u^u $ for all $ \mathbf{l} \in \mathcal{R}^2 $, where $ \lambda_u^l $ and $ \lambda_u^u $ are two positive constants. The distance-based association \cite{XudongZhu:2015} are considered in this paper where users are associated to BSs with minimum distance (pathloss). Let $ \Phi^{\ell} $ be the set of active users in the $ \ell $-th cell and $ |\Phi^{\ell}| $ be the cardinality of set $\Phi^{\ell}$.

The block fading channel is considered, where the channel is quasi-static within one uplink frame and varies in different frames. For the elaboration convenience, we only consider the uplink transmission within coherence bandwidth of one frame. To facilitate the proposed scheme, the uplink frame is divided into $ N + 1$ blocks, where the $ 0 $-th block containing $ L_p $ symbols is for pilot transmission, and the following blocks (from the $1$-st block to the $N$-th one) are for uplink data transmission. Let $B_i$ be the number of symbols in the $i$-th ($i=1,2,...,N$) data block, and $L = \sum\limits_{i = 1}^N {{B_i}}$ be the total length of all data blocks. The \emph{modulation and coding scheme} (MCS) for each data block can be independently adapted. We assume that the length of pilot sequences is sufficiently large and only consider the situation of $ L_p \geq |\Phi^{\ell}|, \forall \ell $, in the analysis.

%\begin{figure}
%	\centering
%	\includegraphics[scale=1.2]{frame.pdf}
%	\caption{Illustration of data-assisted uplink detection scheme.}\label{fig:flowchart}
%\end{figure}

Each user has a unique index in its service cell. The $k$-th user in the $\ell$-th cell is referred to as the $ (\ell,k) $-th user, whose location is denoted by a 2-dimensional vector $ \mathbf{\bf l}_{\ell,k} $. 
The notations on the uplink transmission are listed in the following.
\begin{itemize}
	\item $ \mathbf{x}^0_{\ell,k} \in \mathcal{C}^{1 \times L_p}, \forall (\ell,k) \in \Phi^{\ell}$, denotes the pilot sequence of the $(\ell,k)$-th user. A set of orthogonal pilot sequences are used in each cell for uplink channel estimation. To avoid excessive pilot overhead, the pilot sequences of different cells are not orthogonal. Their cross-correlation satisfies
	\begin{align}
	\frac{|\mathbf{x}^0_{\ell,k} (\mathbf{x}^0_{i,j} )^H|}{L_p}\! =\! \frac{P}{\sqrt{L_p}},\ \ \forall \ell \!\neq\! i, \ (\ell,k) \!\in\! \Phi^{\ell},\ (i,j) \!\in\! \Phi^{i}, \nonumber
	\end{align}
	where $ P $ is the transmission power of mobile users. This assumption matches the LTE systems, where pilot sequences are generated from Zadoff-Chu sequences \cite{TR25996:3GPP}. 
	
	\item $\mathbf{x}^{n}_{\ell,k} \!\!\in\!\! \mathcal{C}^{1 \times B_n}, \forall n\!\!=\!\!1,2,...,N, (\ell,k) \!\!\in\!\! \Phi^{\ell}$, denotes the uplink data symbols in the $ n $-th data block of the $ (\ell,k)$-th user, whose components are complex Gaussian distributed with zero mean and unit variance. $ \mathbf{x}^{m,n}_{\ell,k}\!\!\! =\!\! [\mathbf{x}^m_{\ell,k} \ \mathbf{x}^{m+1}_{\ell,k} \ ... \ \mathbf{x}^n_{\ell,k}]$ ($ 0 \!\leq\!\! m \!\!\leq\!\! n \!\!\leq\!\! N $) is the uplink data symbols of the $ (\ell,k)$-th user from the $ m $-th block to the $ n $-th block,
	and $ \mathbf{X}_\ell^{m,n}\!\! = [(\mathbf{x}_{\ell,1}^{m,n})^T (\mathbf{x}_{\ell,2}^{m,n})^T ... (\mathbf{x}_{\ell,|\Phi^\ell|}^{m,n})^T]^T$, where $ \mathbf{x}^{T} $ is the transpose of $ \mathbf{x} $.
	
	\item  $ \mathbf{h}_{\ell,k} \in \mathcal{C}^{M \times 1} $ denotes the uplink channel vector from the $ (\ell,k)$-th user to the target BS. It is assumed that $ \mathbf{h}_{\ell,k} \in \mathcal{CN}(0,\mathbf{R}_{\ell,k}) $ is complex Gaussian distributed with zero mean and covariance matrix $ \mathbf{R}_{\ell,k}=\mathbb{E}[\mathbf{h}_{\ell,k} \mathbf{h}_{\ell,k}^H]=\rho^1_{\ell,k} \mathbf{I}$, where $ \rho^1_{\ell,k} = \chi^1 _{\ell,k} |\mathbf l_{\ell,k}|^{-\sigma} $ denotes the large-scale fading coefficient from the $ (\ell,k) $-th user to the target BS consisting of log-normal shadowing $\chi^1_{\ell,k}$ and pathloss, $\mathbf{I}$ is an identity matrix, $ \sigma $ is the pathloss exponent. It is assumed that $\chi^1_{\ell,k} = 10 ^{\zeta/ 10} $, where is $\zeta$ a Gaussian random variable with zero mean and standard deviation $\theta$. Similarly, let $ \rho^i_{\ell,k} $ and $ \chi^i_{\ell,k} $ be the large-scale fading coefficient and shadowing from the $ (\ell,k) $-th user to the $ i $-th BS, respectively. In order to simplify the expressions, we shall neglect the superscript $ 1 $ and use $ \rho_{\ell,k} $ and $ \chi_{\ell,k} $ to represent the large-scale fading coefficient and shadowing from the $ (\ell,k) $-th user to the target BS in remaining parts. 
	
	\item $ \mathbf{H}_{\ell}\!\! =\!\! [\mathbf{h}_{\ell,1} \ \mathbf{h}_{\ell,2} \ ... \ \mathbf{h}_{\ell,|\Phi^\ell|}] $, $ \widetilde{\mathbf{R}}_\ell \!\!=\!\! \mathbb{E} [\mathbf{H}^H_\ell \mathbf{H}_\ell] \!\!=\!\! M \!\!\times\!\! \text{diag} \{\rho_{\ell,1}, \rho_{\ell,2}, ..., \rho_{\ell,|\Phi^{\ell}|}\}$ and $ \mathbf{R}_\ell \!\!=\!\! \mathbb{E} [ \mathbf{H}_\ell \mathbf{H}^H_\ell] \!\!=\!\! \sum\limits_{k=1}^{|\Phi^\ell|} \rho_{\ell,k} \mathbf{I}. $
	
\end{itemize}

Without loss of generality, we shall first study the uplink performance of the target cell and extend the conclusion to all the cells. The received symbols of the target BS from the $ m $-th block to the $ n $-th block ($ 0 \leq m \leq n \leq N $) are given by
\begin{equation}
\mathbf{Y}^{m,n} = \mathbf{H}_{1} \mathbf{X}_{1}^{m,n}+ \sum_{\forall \ell \neq 1}  \mathbf{H}_{\ell} \mathbf{X}_\ell^{m,n} + \mathbf{Z}^{m,n}, \ 0 \leq m \leq n \leq N, \nonumber
\end{equation}
where $\mathbf{Z}^{m,n}$ is complex white Gaussian noise with zero mean and variance $\sigma_z^2$ for each element. %In the preliminary work \cite{Upadhya:2017}, we proposed to exploit the uplink frame structure with independent coding blocks via an iterative channel estimation and data detection scheme. In the next section, we extend the scheme by exploiting the BS cooperation, where detected data blocks of neighboring cells can be shared via backhaul. 

\section{Proposed Uplink Scheme and Analysis} \label{sec:scheme}

The cooperative uplink detection scheme is presented below.

\begin{Scheme}[Uplink Detection with Decoded Interference] 
	We only elaborate the uplink detection steps of the target cell for convenience, which is also applied on all other cells. 
	\begin{itemize}
		\item {\bf Step I:} Initialize the iteration index by $ i=0 $.
		
		\item {\bf Step II:} In the $ i $-th iteration, the target BS estimates the uplink channels of its active users according to the pilot sequences $ \mathbf{X}_1^0 $ and the detected data symbols of the previous iterations $ \mathbf{X}_1^{1,i} $. The estimated channel is denoted as $ \widehat{\mathbf{H}}_{1}^i=[\widehat{\mathbf{h}}_{1,1}^i \quad
		\widehat{\mathbf{h}}_{1,2}^i \quad ... \quad \widehat{\mathbf{h}}_{1,|\Phi^{1}|}^i ]=\mathbf{Y}^{0,i} \mathbf{Q}^{i}_{in} $. The estimator $ \mathbf{Q}^i_{in} $ based on minimum-mean-square-error (MMSE) criterion is given as
		\begin{align}
		\mathbf{Q}^{i}_{in} \!&= [\mathbf{q}^i_1 \ \mathbf{q}^i_2 \ ... \ \mathbf{q}^i_{|\Phi^{1}|}]\nonumber\\
		&=\arg\min_{\mathbf{Q}^{i}_{in}} tr \bigg\{ \mathbb{E} \bigg[(\widehat{\mathbf{H}}_{1}^i - \mathbf{H}_{1})^{H} (\widehat{\mathbf{H}}_{1}^i - \mathbf{H}_{1}) \bigg| \mathbf{X}_1^{0,i} \bigg] \bigg\} \nonumber\\
		\!&= \!\!\bigg[\!(\mathbf{X}_1^{0,i})^{\!H}\! \widetilde{\mathbf{R}}_{1}\! \mathbf{X}_1^{0,i} \!\!+\!\! P \!\!\sum_{\forall \ell\neq 1}\!\! tr(\widetilde{\mathbf{R}}_{\ell}) \mathbf{I} \!\!+\!\! |\Phi^{1}|\mathbf{I} \!\bigg]^{\!\!-1}\!\!\!\!\!\! (\mathbf{X}_1^{0,i})^{\!H}\! \widetilde{\mathbf{R}}_{1},\nonumber
		\end{align}
		where the expecation is take over all possible $\mathbf{Y}^{0,i}$.
		
		\item {\bf Step III:} If $ i \leq d-1 $, the target BS detects the uplink data in the $ (i+1) $-th block, denoted as $ \widehat{\mathbf{X}}^{i+1}_1 $, according to the latest channel estimation $ \widehat{\mathbf{H}}_{1}^i $ with
		\begin{equation}
		\widehat{\mathbf{h}}_{1,k}^{i} = \mathbf{h}_{1,k} +  \underbrace{\mathbf{h}_{1,k} \bigg( \mathbf{x}_{1,k}^{0,i} \mathbf{q}_{k}^i -1\bigg) + \sum\limits_{\forall (\ell,j) \neq (1,k)} \mathbf{h}_{\ell,j} \mathbf{x}_{\ell,j}^{0,i} \mathbf{q}_{k}^i}_{\mbox{Denoted as }\Delta \mathbf{h}_{1,k}^{i}},\nonumber
		\end{equation}
		where the MMSE-based data detector is derived as
		\begin{align}
		\mathbf{S}_{in}^{i+1}\!\!&=\arg\min_{\mathbf{S}_{in}^{i+1}} tr \bigg\{\! \mathbb{E}\! \bigg[\! \widehat{\mathbf{X}}_1^{i+1}\!\!-\!\!\mathbf{X}_1^{i+1} \bigg]^{\!\!H}\!\! \bigg[ \widehat{\mathbf{X}}_1^{i+1}\!\!-\!\!\mathbf{X}_1^{i+1} \bigg| \widehat{\mathbf{H}}_1^i \bigg]\!\bigg\} \nonumber\\
		&= \!(\widehat{\mathbf{H}}_{1}^{i})^{\!H} \!\!\underbrace{\!\bigg[\! \widehat{\mathbf{H}}_{1}^{i}\! (\widehat{\mathbf{H}}_{1}^{i})^{\!H} \!\!+ \!\! \sum_{\forall j}\!\! \Delta \mathbf{R}_{1,j}^i  \!\!+\! \!\!\!\sum_{\forall \ell\neq 1, k} \!\!\!\!\mathbf{R}_{\ell,j} \!\!+\!\! \frac{\mathbf{I}}{P} \!\bigg]^{\!\!-1}}_{\mbox{Denoted as }\mathbf{\Psi}^{i}_{in}}\!\!.
		\end{align}
		The BS decodes the information bits in $\widehat{\mathbf{X}}_1^{i+1}$, and reconstructs $ \mathbf{X}^{i+1}_1 $, which will be used as pilots in the next iteration. Let $ i=i + 1 $, and jump to Step II. 
		
		If $ i \ge d $, before data detection, the target BS uses $ \widehat{\mathbf{H}}_{1}^i $ to cancel the detected signals from the received signals from the pilot sequence to the $(i-d)$-th block as follows
		\begin{align}
		\mathbf{Y}^{0,i-d}_{intf} =& \mathbf{Y}^{0,i-d} - \widehat{\mathbf{H}}_{1}^i \mathbf{X}_1^{0,i-d}\nonumber\\
		 =& \sum_{\forall \ell \neq 1} \mathbf{H}_{\ell} \mathbf{X}_\ell^{0,i-d} \!-\! \Delta\mathbf{H}_{1}^i\mathbf{X}_{1}^{0,i-d} + \mathbf{Z}^{0,i-d},
		\end{align}
		where $ \Delta\mathbf{H}_{1}^i = \widehat{\mathbf{H}}^i_1 - \mathbf{H}_1 $.
		
		\item {\bf Step IV:} Estimate the channel between the users in $ \Phi_{co}^1 $ and the target BS, denoted as $ \mathbf{H}_{intf} =[\mathbf{h}_{\ell,j}]_{(\ell,j) \in \Phi_{co}^1}$, from $ \mathbf{Y}^{0,i-d}_{intf} $ according to the MMSE channel estimation. The estimated channel $ \widehat{\mathbf{H}}_{intf}^i $ and the estimator $ \mathbf{Q}_{co}^i $ are given as
		\begin{eqnarray}
		\widehat{\mathbf{H}}^{i}_{intf}=[\widehat{\mathbf{h}}^{i}_{\ell,j}]_{(\ell,j)\in \Phi_{co}^1}  =\mathbf{Y}^{0,i-d}_{intf} \mathbf{Q}^{i}_{co},\nonumber
		\end{eqnarray}
		\begin{align}
		\mathbf{Q}^{i}_{co} \!=& \bigg[(\mathbf{X}_{intf}^{0,i-d})^{\!H} \widetilde{\mathbf{R}}_{intf} \mathbf{X}_{intf}^{0,i-d} \!\!+\! M\! P\!\!\!\! \sum_{\forall (\ell,j) \atop \notin \Phi_{co}^1 \cup \Phi^{1}}\!\!\!\! \rho_{\ell,j} \mathbf{I} \!+\! |\Phi_{co}^1|\mathbf{I} \nonumber\\
		\!\!+&(\mathbf{X}_{1}^{0,i-d})^{\!H} \!\mathbb{E} [(\!\Delta\! \mathbf{H}_1^i\!)^{\!H}\!\! \Delta\! \mathbf{H}_1^i] \mathbf{X}_{1}^{0,i-d}\!\bigg]^{\!\!-1} \!\!\!\!\!
		(\mathbf{X}_{intf}^{0,i-d})^{\!H} \!\widetilde{\mathbf{R}}_{intf} \nonumber,
		\end{align}
     	where $	\mathbb{E} [(\Delta \mathbf{H}_1^i)^H \Delta \mathbf{H}_1^i]= \bigg[\widetilde{\mathbf{R}}_1^{-1} + \frac{\mathbf{X}_1^{0,i} (\mathbf{X}_1^{0,i})^H}{P \sum_{\forall \ell \neq 1} tr(\widetilde{\mathbf{R}}_\ell)} \bigg]^{-1}$,
		$\widetilde{\mathbf{R}}_{intf} = \mathbb{E} [(\mathbf{H}_{intf})^H \mathbf{H}_{intf}]$, and $ \mathbf{X}_{intf}^{0,i-d} $ is the matrix aggregating the signals from pilot sequence to the $(i-d)$-th uplink block of users in $ \Phi_{co}^1 $.
		
		\item {\bf Step V:} The target BS detects the uplink data in the $ (i+1) $-th block, denoted as $ \widehat{\mathbf{X}}^{i+1}_1 $, according to the latest channel estimation $ \widehat{\mathbf{H}}_{1,co}^i = [\widehat{\mathbf{H}}_{1}^i, \widehat{\mathbf{H}}_{intf}^i] $. The MMSE-based data detector is given as
		\begin{align}
		\mathbf{S}^{i+1}_{co} = &(\widehat{\mathbf{H}}_{1,co}^i)^{H}\! \bigg[ \widehat{\mathbf{H}}_{1,co}^i (\widehat{\mathbf{H}}_{1,co}^i)^{H} \!\!\!+\!\! \!\!\!\sum\limits_{\forall (\ell,j) \atop \notin \Phi_{co}^1 \cup \Phi^{1}}\!\!\!\!\! \mathbf{R}_{\ell,j}
		\!+\!\! \sum\limits_{\forall k}\!\! \Delta \mathbf{R}^i_{1,k} \nonumber\\
		&+ \sum\limits_{\forall (\ell,j)\atop \in \Phi_{co}^1} \Delta \mathbf{R}^i_{\ell,j} + \mathbf{I}/P \bigg]^{-1}
		= (\widehat{\mathbf{H}}_{1,co}^i)^H \mathbf{\Psi}^{i}_{co}. \nonumber
		\end{align}
		If $ i<N-1$, the BS reconstructs $ \mathbf{X}^{i+1}_1 $, which will be used as pilots in the next iteration, and the algorithm jumps to Step II; otherwise, the iteration terminates.
	\end{itemize}
	\label{sch:scenario2}
\end{Scheme}

\begin{figure}[!htpb]
	\begin{minipage}[t]{0.45\linewidth}%设定图片下字的宽度，在此基础尽量满足图片的长宽
		\centering
		\includegraphics[height = 90pt]{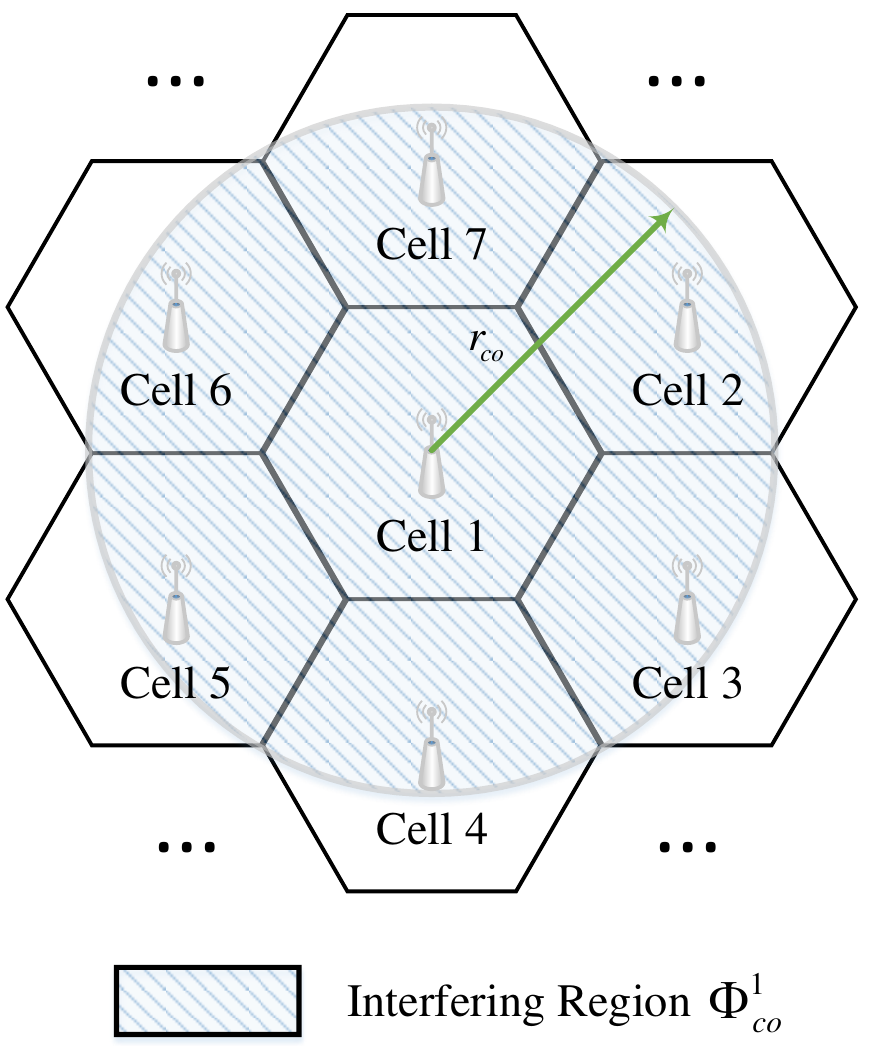}
		\caption*{(a) BS coorperation in an hexagonal network.}%加*可以去掉默认前缀，作为图片单独的说明
		\label{fig:side:a}
	\end{minipage}
	\begin{minipage}[t]{0.5\linewidth}%需要几张添加即可，注意设定合适的linewidth
		\centering
		\includegraphics[height = 95pt]{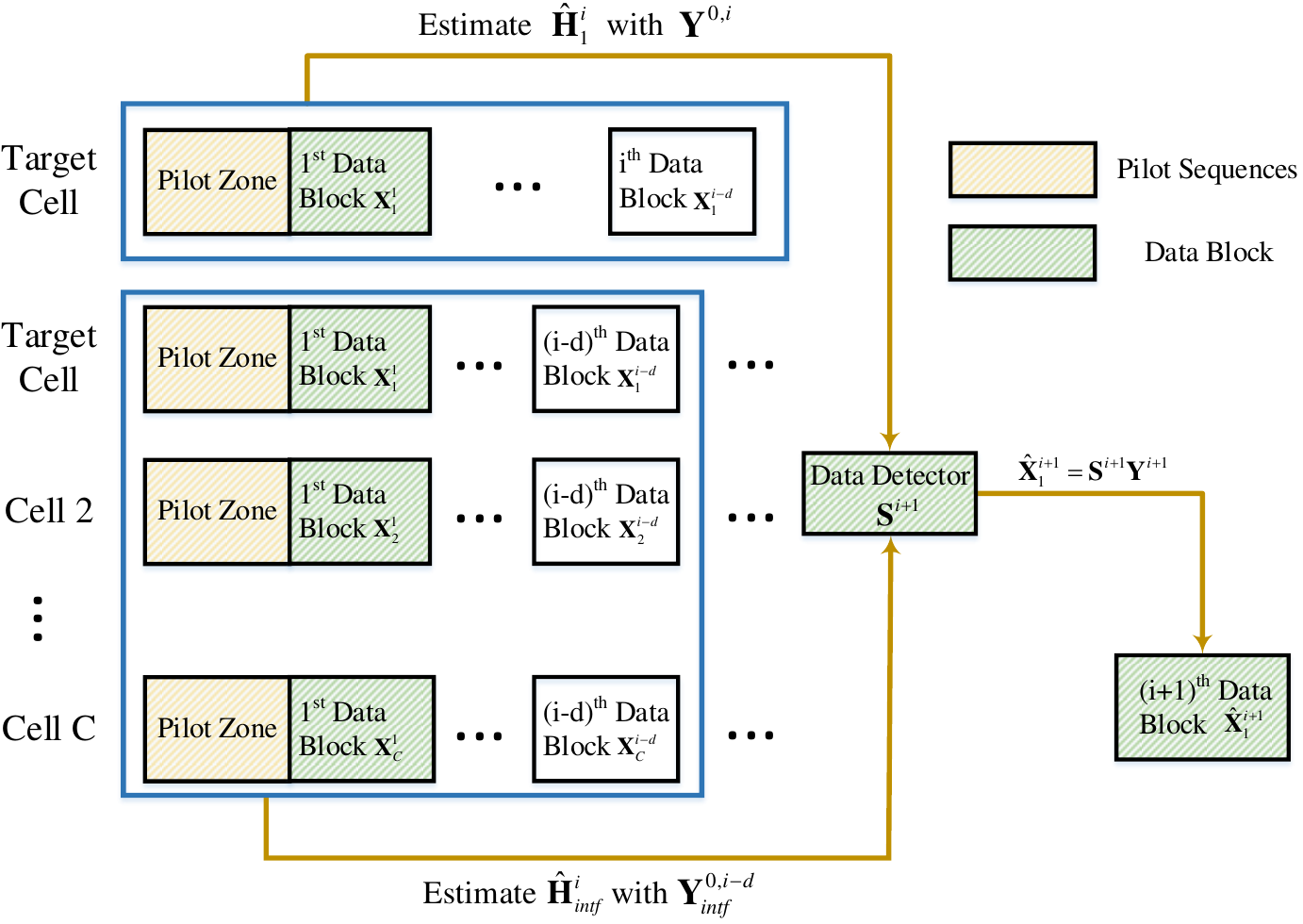}
		\caption*{(b) The detection procedure of the $(i+1)$-th data block ($i \ge d$) in the target cell.}
		\label{fig:side:b}
	\end{minipage}
	\caption{The proposed data-assisted transmission and detection scheme via BS cooperation.}\label{proposed}
\end{figure}

%Note that $ \mathbf{x}_{1,m}^{0,i} $ ($ \forall m $) is known to the target BS, the MMSE channel estimator $ \mathbf{q}_{k}^i $ could effectively mitigate the intra-cell interference terms ($ \mathbf{x}_{1,m}^{0,i}$ and $\mathbf{q}_{k}^i $ are almost orthogonal for $ m\neq k $). On the other hand, the inter-cell interference depends on the cross-correlation between  $ \mathbf{x}_{\ell,j}^{0,i}$ ($ \ell\neq 1 $) and $ \mathbf{q}_{k}^i  $. With larger number of symbols in $ \mathbf{x}_{\ell,j}^{0,i}$, such cross-correlation becomes smaller. In the $ i $-th iteration, when $ i \leq d-1 $, the data block $ \mathbf{X}_1^{i+1} $ could be detected according to the estimated channel matrix $ \widehat{\mathbf{H}}^i_1 $. The MMSE detector is given below \cite{Hoydis:13}

%where the expectation is taken over all possible data symbols and channel estimation error, and $\Delta \mathbf{R}_{1,k}^i =  \mathbb{E} [\Delta \mathbf{h}_{1,k}^{i} (\Delta \mathbf{h}_{1,k}^{i})^{H}]$. 

%\begin{Remark}
%	Note that users in the set of $ \Phi_{co} $ should be aware at the neighboring BSs. This can be achieved if users keep tracking the pathloss to all the BSs of neighboring cells, which is also necessary for the preparation of handover. Alternatively, the detection of set $ \Phi_{co} $  can also be accomplished via positioning techniques, like Global Positioning System. Moreover, the determination of $ \Phi_{co} $ can also be made by comparing large-scaling fading coefficient (instead of pathloss) with a threshold. The analysis is similar and it is omitted here due to page limitation.
%\end{Remark}

\subsection{Asymptotic Uplink SINR}\label{sub:asp-sinr}

Since the active users are independently distributed, we only provide the uplink SINR analysis of the first active user in the target cell. The first user of the target cell is called as the target user in the remainder of this paper. With the proposed scheme, the uplink SINR of the target user is then given by (\ref{eqn:ul-sinr}) and (\ref{eqn:sir-sc2}) for $i \le d-1$ and $i \ge d$, respectively. In both SINR expressions, the terms I, II and III are due to the intra-cell interference, channel estimation error, and inter-cell interference, respectively. The expectation is taken over all possible small-scale fading in other cells.

\newcounter{mytempeqncnt}
\begin{figure*}[!t]
	\normalsize
	\setcounter{equation}{2}
\begin{equation}\small
\gamma_{1,1}^{i+1} = \frac{|(\widehat{\mathbf{h}}_{1,1}^{i})^H \mathbf{\Psi}^{i}_{in} \widehat{\mathbf{h}}_{1,1}^{i}|^2}{\underbrace{\sum_{\forall k \neq 1}|(\widehat{\mathbf{h}}_{1,1}^{i})^H \mathbf{\Psi}^{i}_{in} \widehat{\mathbf{h}}_{1,k}^{i}|^2}_{\mbox{Interference I}} + \mathbb{E} \bigg[ \!
	\underbrace{\sum_{\forall k} | (\widehat{\mathbf{h}}_{1,1}^{i})^H \mathbf{\Psi}^{i}_{in} \Delta \mathbf{h}_{1,k}^{i}|^2}_{\mbox{Interference II}}+ \!\!\underbrace{\sum_{\forall \ell \neq 1, j} | (\widehat{\mathbf{h}}_{1,1}^{i})^H \mathbf{\Psi}^{i}_{in} \mathbf{h}_{\ell,j}^i|^2}_{\mbox{Interference III}} \bigg| \widehat{\mathbf{H}}_{1}^{i} \bigg] \!+\! \underbrace{ \frac{\left\|(\widehat{\mathbf{h}}_{1,1}^{i})^H \mathbf{\Psi}^{i}_{in}\right\|^2}{P}}_{\mbox{Noise}}}
\label{eqn:ul-sinr}
\end{equation}
\begin{align}\small
\gamma_{1,1}^{i+1} \!=\!\! \frac{|(\widehat{\mathbf{h}}_{1,1}^{i})^H \mathbf{\Psi}^{i}_{co} \widehat{\mathbf{h}}_{1,1}^{i}|^2}{\underbrace{\sum_{{(\ell,j) \neq (1,1),} \atop {(\ell,j)\in \Phi^{1} \cup \Phi_{co}^1}  }\!\!|(\widehat{\mathbf{h}}_{1,1}^{i})^H \mathbf{\Psi}^{i}_{co} \widehat{\mathbf{h}}_{\ell,j}^{i}|^2}_{\mbox{Interference I}} \!+ \mathbb{E} \bigg[ \!\!\!\!
	\underbrace{\sum_{\forall (\ell,j) \atop \in \Phi^{1} \cup \Phi_{co}^1}\!\!\!\! |(\widehat{\mathbf{h}}_{1,1}^{i})^H \mathbf{\Psi}^{i}_{co} \Delta \mathbf{h}_{\ell,j}^{i}|^2}_{\mbox{Interference II}} \!+\!\!\!\!\!
	\underbrace{\sum_{\forall (\ell,j) \atop \notin \Phi^{1} \cup \Phi_{co}^1} \!\!\!\!|(\widehat{\mathbf{h}}_{1,1}^{i})^H \mathbf{\Psi}_{co}^{i} \mathbf{h}_{\ell,j}|^2}_{\mbox{Interference III}} \bigg| \widehat{\mathbf{H}}_{1,co}^{i}\!\bigg] \!\!+\! \underbrace{ \frac{\left\|(\widehat{\mathbf{h}}_{1,1}^{i})^H \!\mathbf{\Psi}^{i}_{co}\right\|^2}{P}}_{\mbox{Noise}}} . \label{eqn:sir-sc2}
\end{align}
	\hrulefill
	\vspace*{0pt} 
\end{figure*}
\setcounter{equation}{4}

Then, the asymptotic expressions of the uplink SINR when there are sufficient data symbols involved in channel estimation are given by the following lemma.
\begin{Lemma}[Asymptotic Expressions of SINR] \label{lem:sce1}
	Let $\Phi_{eff} = \Phi^{1} \!\cup \Phi_{co}^1$, $L_i \!\!=\! \!L_p \!+\!\! \sum\limits_{m = 1}^i\!\! {{B_m}}$ and $ L_i^{\prime} \!\!=\!\! L_p \!+\!\! \sum\limits_{m = 1}^{i - d}\!\! {{B_m}}  $. For sufficiently large $ M $ and $L_i^{\dagger}$, the uplink SINR of the target user is
	\begin{eqnarray}
	\gamma_{1,1}^{i+1} &\rightarrow& \frac{1}{(\frac{{\left| {{\Phi^\dagger _i}} \right|}}{{{L_i^\dagger}}} + 1)\sum\limits_{\forall (\ell,j) \notin \Phi^\dagger _i} \frac{1}{M} \frac{\rho_{\ell,j}}{\rho_{1,1} } + \frac{1}{{L_i^\dagger}} (\frac{\rho_{\ell,j}}{\rho_{1,1} })^2 }, \label{eqn:lemma_sc1}
	\end{eqnarray}
	where $ \Phi^\dagger_i = {\Phi ^{1}}$, ${{L_i^\dagger}} = L_i$ for $ i \leq d-1 $ and  $ \Phi^\dagger_i = {\Phi _{eff}}$, ${{L_i^\dagger}} = {L_i^\prime}$ for $ i \ge d $.
\end{Lemma} 

\begin{proof}
	Please refer to Appendix A.
\end{proof}

From Lemma \ref{lem:sce1}, we can observe that the uplink interference in (\ref{eqn:lemma_sc1}) is due to the users outside target cell $ \Phi^{1} $ when $i \le d-1$, and the users outside $ \Phi^{1} \cup \Phi_{co}^1 $ when $i \ge d+1$. Thus BS cooperation will lead to better performance. It can be observed that the SINR of the target user is a function of the locations of interfering users, which are usually unknown to each service BS. Hence, the accurate value of uplink SINR is hard to predict before transmission. Considering the randomness in active user distribution, we continue to analyze the distribution of the asymptotic SINR in the following subsections, which is necessary for robust rate allocation with a target outage probability.

\subsection{Asymptotic SINR Distribution} \label{sub:distance-2}

The asymptotic CDF expression of the SINR distribution is given by the following theorem.

\begin{Theorem}[Asymptotic CDF of SINR for Distance-based Cell Association] \label{theo:sc1_1}
	For sufficiently large $ L_i^\dagger $ and $ M $, the CDF of the target user's SINR $\gamma_{1,1}^{i+1}$ can be written as 
	\begin{equation}
	\Pr\left[ \gamma_{1,1}^{i+1} < T  \right] \rightarrow Q \bigg[ \bigg(\frac{1}{T (\frac{{ {|{\Phi^\dagger _i}|} }}{{{L_i^\dagger}}} + 1) } - \mathcal{M}_{i+1} \bigg) \sqrt{\frac{1}{\mathcal{V}_{i+1}}} \bigg], \label{eqn:cdf_s1}
	\end{equation}
	where the $ Q $-function is the tail probability of standard normal distribution. $ \mathcal{M}_{i+1} $ and $ \mathcal{V}_{i+1}$ are the mean and variance of $\sum\limits_{\forall \left( {\ell,j} \right) \notin \Phi^\dagger _i } \!\!\left[ {\frac{{{\rho _{\ell,j}}}}{{M{\rho _{1,1}}}} + \frac{{\rho _{\ell,j}^2}}{{{L_i^\dagger} \rho _{1,1}^2}}}\right] $ respectively, which are given by
	\begin{equation}\label{eqn:s1-M}
	\mathcal{M}_{i+1} =  \int\limits_{\bar S_1}\!\! \left[\!\frac{e^{\frac{a^2\theta^2}{2}}}{M}\frac{{|\mathbf{l}|^{-\sigma}}}{{\rho}_{1,1}} \!+\! \frac{e^{2a^2\theta^2}}{{L_i^\dagger}}{(\frac{{{|\mathbf{l}|^{-\sigma} }}}{{\rho}_{1,1} })^2} \!\right]\!\!{\lambda_u(\mathbf{l})} ds(\mathbf{l}),
	\end{equation}
	\begin{equation}\label{eqn:s1-V}
	\mathcal{V}_{i+1} \!=\!\! \int\limits_{{\bar S_1}} \!\!
	{\left[\frac{e^{\frac{a^2\theta^2}{2}}}{M}\frac{{{|\mathbf{l}|^{-\sigma}}}}{{\rho}_{1,1}} \!+\! \frac{e^{2a^2\theta^2 |\mathbf{l}|^{-2\sigma}}}{{L_i^\dagger} {\rho}_{1,1} ^{2}}\right]^{\!2} \!\!\!{\lambda_u(\mathbf{l})} ds(\mathbf{l})}\!-\! \mathcal{M}_{i+1}^2.
	\end{equation}
	Moreover, $a = \frac{\ln 10} {10}$ and $ |\mathbf{l}| $ is the Euclidean norm of 2-dimensional vector $ \mathbf{l} $. $\bar S_1$ denotes the network region except the target cell for $ i \leq d-1 $, and the region where the distance to the target BS is larger than $r_{co}$ for $ i \ge d $. 
\end{Theorem}

\begin{proof}
		Please refer to Appendix B.
\end{proof}

With the knowledge of $\lambda_u$, ${\cal M}_{i+1}$ and ${\cal V}_{i+1}$ can be calculated numerically using \eqref{eqn:s1-M} and \eqref{eqn:s1-V}, and the corresponding distribution of SINR can be obtained. In fact, a more practical way is to learn these two parameters in real time since $\lambda_u({\bf l})$ is usually unknown at the BSs. The online learning method is elaborated in the following section.

\section{Uplink Rate Adaptation}\label{sec:learning}

Note that the uplnik SINR  expression derived in (\ref{eqn:lemma_sc1}) is the function of the large-scale fading coefficients of interfering channels, which is unknown to each service BS. For example, $ \rho_{\ell,j}$ ($ \forall \ell \neq 1$) are usually not measured at the target BS due to large signaling overhead. In fact, not only the accurate value of uplink SINR, but also its distribution derived in Theorem \ref{theo:sc1_1} are hard to predict due to the lack of knowledge on density $ \lambda_{u} $. This raises an issue on robust uplink rate allocation, which will be addressed in the following Section \ref{sub:rate-allocaion} via an online learning algorithm.

\subsection{Learning-based Uplink Rate Adaptation} \label{sub:rate-allocaion}

According to Theorem \ref{theo:sc1_1}, given a target outage probability $ \varepsilon $, the uplink data rate of the target user in the $i$-th data block for both cell association criteria can be scheduled as 
\begin{equation}\label{sch_rate}
r^{i}_{1,1} = \log_2 \left(1 + T^i_{1,1}(\varepsilon) \right), i = 1,2,...,N,
\end{equation}
where $
T^i_{1,1}(\varepsilon) = \frac{1}{(Q^{-1}(\varepsilon)\sqrt{\mathcal{V}_{i}}+ \mathcal{M}_{i})(\frac{{ |\Phi^\dagger_{i-1}| }}{{{L_{i-1}^\dagger}}} + 1)}$, and $Q^{-1}(\cdot)$ represents the inverse of $Q$ function.

Since the user density $ \lambda_{u} $ is unknown, the statistics $\mathcal{M}_{i}$ and $\mathcal{V}_{i}$ in the above equation cannot be directly calculated. We introduce the following online learning algorithm, which collects the information of user distribution and finally converges to the accurate values of  $\mathcal{M}_{i}$ and $\mathcal{V}_{i}$ for the target user.
\begin{Algorithm}[Learning Algorithm for ${\cal M}_i$ and ${\cal V}_i$]\label{Alg_Learning}
	$\{{\cal M}_i, {\cal V}_i | i=1,...,N\}$ can be evaluated iteratively in each frame, where the target user is scheduled in uplink transmission, by the following steps.
\begin{itemize}
	\item {\bf Step 1:} Initialize the values of $ \mathcal{M}_{i} $ and $ \mathcal{V}_{i}$ ($i=1,2,...,N$), denoted as $ \mathcal{M}^0_{i} $ and $ \mathcal{V}^0_{i}$. Let $n = 1$.
	
	\item {\bf Step 2:} In each data block, schedule one silent symbol in the target cell such that no users of the target cell transmits uplink signal. Let $ I_n^i $ be the $n$-th measured uplink interference power in the $i$-th data block to the target user after the processing of data detector $ \mathbf{S}^{i} $. Then, update $ \mathcal{M}_i $ and $ \mathcal{V}_i$ as 
	\begin{equation}
	\mathcal{M}^n_{i} = \frac{n-1}{n} \mathcal{M}^{n-1}_{i} + \frac{1}{n} \frac{I_n^i}{\frac{|\Phi^\dagger_{i-1}|}{{{L_{i-1}^\dagger}}} + 1}, \label{eqn_M_Leran}
	\end{equation}
	\begin{equation}
	\mathcal{V}^n_{i} = \frac{n-2}{n-1} \mathcal{V}^{n-1}_{i} + \frac{1}{n-1} (\frac{I_n^i}{\frac{|\Phi^\dagger_{i-1}|}{{{L_{i-1}^\dagger}}} + 1} - \mathcal{M}^{n-1}_{i})^2. \label{eqn_v_Leran}
	\end{equation}
	\item {\bf Step 3:} Let $n=n+1$ and repeat Step 2 in the frames where the target user is scheduled for uplink transmission, until the iteration converges.
\end{itemize}
\end{Algorithm}
Since $ \frac{I_n}{\frac{|\Phi^\dagger_{i-1}|}{{{L_{i-1}^\dagger}}} + 1} $ is an unbiased observation of $ \mathcal{M}_i $, it  is easy to see that $(
\mathcal{M}_{i}^n, \mathcal{V}_{i}^n) \rightarrow (\mathcal{M}_{i}, \mathcal{V}_{i})$ for $i=1,...,N$ when $ n \rightarrow + \infty$, which is also shown numerically in Section \ref{sec:sim}. 

The above learning algorithm can be applied on all active users to facilitate the robust uplink rate allocation. Let $ \Phi^{\dagger}_i (b) = \Phi^b $ for $ i \leq d-1 $ and $ \Phi^{\dagger}_i (b) = \Phi^b \cup \Phi^b_{co} $ for $ i \geq d $, and
\begin{eqnarray}\label{DefM}
\mathcal{M}_i(b,k) =\mathbb{E} \sum\limits_{\forall (\ell,n) \notin \Phi^\dagger_{i-1}(b)}  \left[  \frac{1}{M} \frac{\rho_{\ell,n}^b}{\rho_{b,k}^b } + \frac{1}{L_{i-1}^\dagger} (\frac{\rho_{\ell,n}^b}{\rho_{b,k}^b })^2\right]
\end{eqnarray}
\begin{eqnarray}\label{DefV}
\mathcal{V}_i(b,k) = Var  \sum\limits_{\forall (\ell,n) \notin \Phi^\dagger_{i-1}(b)}  \left[  \frac{1}{M} \frac{\rho_{\ell,n}^b}{\rho_{b,k}^b } + \frac{1}{L_{i-1}^\dagger} (\frac{\rho_{\ell,n}^b}{\rho_{b,k}^b })^2\right]
\end{eqnarray}
be the statistical parameters for the $ (b,k) $-th user ($\forall b,k$), then the SINR threshold can be written as $ T^i_{b,k} (\epsilon) =  \frac{1}{(Q^{-1}(\varepsilon)\sqrt{\mathcal{V}_{i}(b,k)}+ \mathcal{M}_{i}(b,k))(\frac{{ |\Phi^\dagger_{i-1}(b)| }}{{{L_{i-1}^\dagger}}} + 1)}$, where $ \mathcal{M}_i(b,k) $ and $ \mathcal{V}_i(b,k) $ can be learned with Algorithm \ref{Alg_Learning}. In fact, if $ \mathcal{M}_i $ and $ \mathcal{V}_i $ have been learned at three users of one cell, their values for other users of the same cell can be calculated directly as follows.
\begin{Lemma}\label{Ext_M&V}
	Suppose that $ \mathcal{M}_i $ and $ \mathcal{V}_i $ ($ \forall i $) have been learned for the $ (b,m) $-th, $ (b,j) $-th and $(b,t)$-th users, denoted as ${\cal M}_i(b,m)$, ${\cal V}_i(b,m)$, ${\cal M}_i(b,j)$, ${\cal V}_i(b,j)$ and ${\cal M}_i(b,t)$, ${\cal V}_i(b,t)$ respectively. For the arbitrary $(b,k)$-th user, we have
	\begin{align}\label{DeduceM}
	{{\cal M}_i}(b,k) =& \frac{{{\cal M}_i^2(b,m)(\rho _{b,m}^b)^2 - {\cal M}_i^2(b,j)(\rho _{b,j}^b)^2}}{{{\rho _{b,k}^b}({\rho _{b,m}^b} - {\rho _{b,j}^b})}} \nonumber\\
	&+ \frac{{{\rho _{b,m}^b}{\rho _{b,j}^b}({{\cal M}_i}(b,m){\rho^b_{b,m}} \!\!-\!\! {{\cal M}_i}(b,j){\rho ^b_{b,j}})}}{{(\rho _{b,j}^b)^2({\rho^b_{b,j}} - {\rho^b_{b,m}})}},
	\end{align}
%Moreover, the variance $\mathcal{V}_i(b,k)$ is given by 	
\begin{align}
	{V_i}(b,k) = \left[ {\begin{array}{*{20}{c}}
		{\frac{1}{{{(\rho^b_{b,m})}^2}}}&{\frac{1}{{{(\rho^b_{b,m})}^3}}}&{\frac{1}{{{(\rho^b_{b,m})}^4}}}
		\end{array}} \right]{{\bf{U}}^{ - 1}}{\bf{V}} - {M_i}^2(b,k),\nonumber
\end{align}
where
\begin{align}
{\bf{U}} \!\!= \!\!\left[\!\!\! {\begin{array}{*{20}{c}}
	{\frac{1}{{{(\rho^b_{b,m})}^2}}}&\!\!{\frac{1}{{{(\rho^b_{b,m})}^3}}}&\!\!{\frac{1}{{{(\rho^b_{b,m})}^4}}}\\
	{\frac{1}{{{(\rho^b_{b,j})}^2}}}&\!\!{\frac{1}{{{(\rho^b_{b,j})}^3}}}&\!\!{\frac{1}{{{(\rho^b_{b,j})}^4}}}\\
	{\frac{1}{{{(\rho^b_{b,t})}^2}}}&\!\!{\frac{1}{{{(\rho^b_{b,t})}^3}}}&\!\!{\frac{1}{{{(\rho^b_{b,t})}^4}}}
	\end{array}} \!\!\!\right], {\bf{V}}\!\! =\!\! \left[\!\!\! {\begin{array}{*{20}{c}}
	{{V_i}(b,m) + {M_i}^2(b,m)}\\
	{{V_i}(b,j) + {M_i}^2(b,j)}\\
	{{V_i}(b,t) + {M_i}^2(b,t)}
	\end{array}} \!\!\!\right]\nonumber
\end{align}
\end{Lemma}

\begin{proof}
	Please refer to Appendix C.
\end{proof}

\section{Simulation and Discussion}\label{sec:sim}
In this section, we demonstrate the performance of the proposed scheme by numerical simulations, and corroborate the analytical results derived in the previous sections. As a comparison, we also show the simulation results for the data-assisted uplink scheme in \cite{Ray:15} (denoted as the baseline scheme), where there is no BS cooperation. In the simulation, a hexagonal cellular network is considered. The cell radius is $ R=500$ m, and the pathloss exponent is $\sigma=3.76$. The standard deviation of shadowing is  $\theta = 3$ dB. The number of BS antennas is $M=200$. We assume $P = 23$ dBm, the thermal noise density is $-174$ dBm/Hz, and the bandwidth is $5$ MHz. Let $ r_{1,1}=|{\bf l}_{1,1}| $ be the distance from the target user to the target BS. We consider the performance of the target user when it is at the cell center ($ r_{1,1} = 100$ m), middle distance ($ r_{1,1} = 300$ m), and cell edge ($ r_{1,1} = 400$ m), respectively. The pilot is the shifted Zadoff-Chu sequence used in the LTE systems. $r_{co}$ is set to be $ 700$ m for the proposed scheme. The uplink data symbols are divided into $5$ blocks, each with $100$ data symbols, and the pilot length $ L_p$ is $ 31 $.

In Fig. \ref{fig:comp_scm2}, both analytical performance and numerical results are demonstrated for the proposed scheme, where the backhaul delay $d=1$. The average number of users per cell is $ 10 $. The CDFs of SINR distributions for different locations of the target user in the last ($5$-th) data block are plotted. In Fig. \ref{fig:comp_scm2}, the analytical performance is generated from \eqref{eqn:cdf_s1}. We can observe that the analytical results fit the numerical results tightly.

\begin{figure}
	\centering
	\includegraphics[height = 4cm, width=5.6cm]{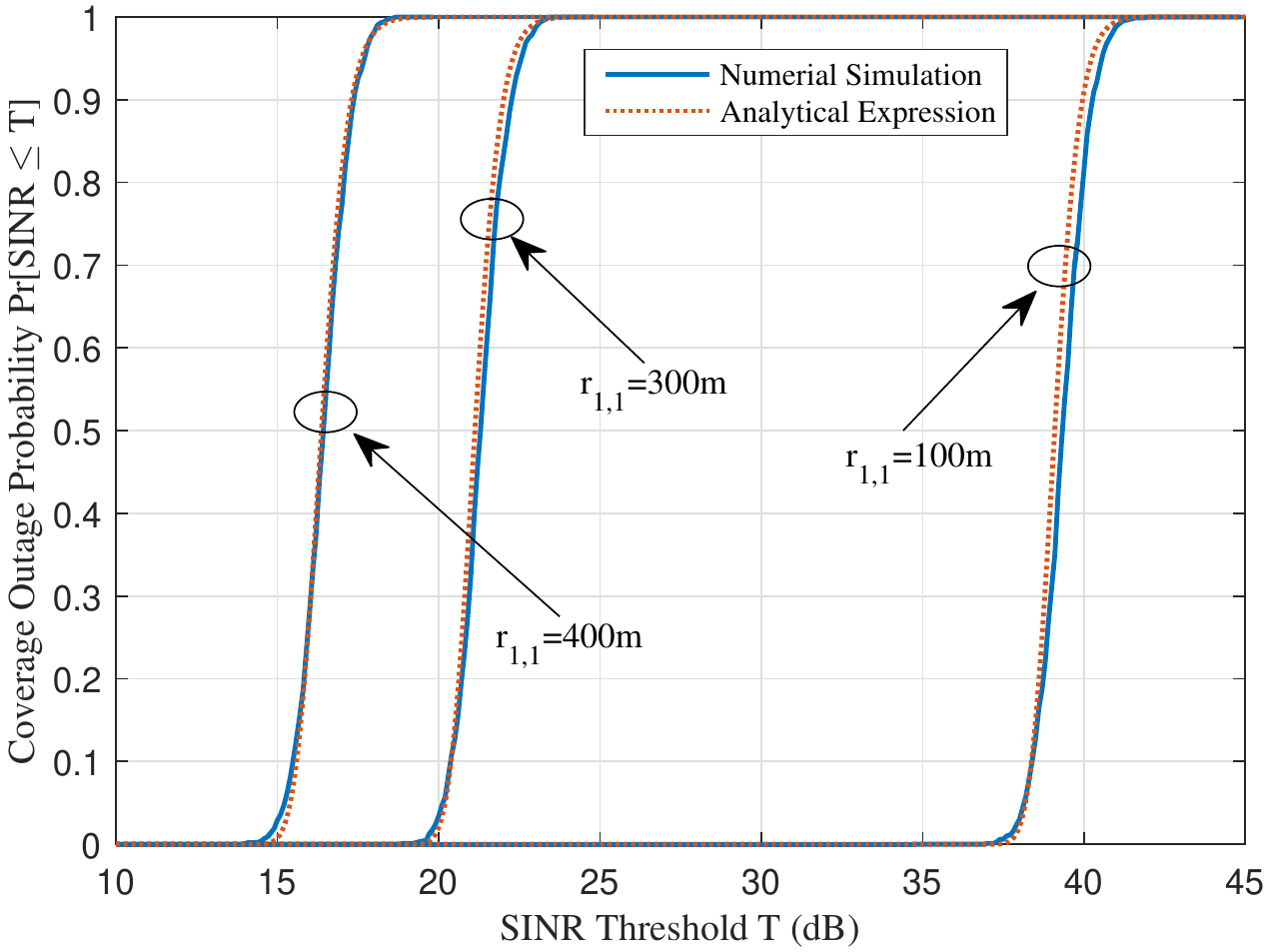}
	\caption{CDFs of the uplink SINR at the $5$-th block for the proposed scheme,  where $ M\!\!\!=\!\!200$, $ r_{1,1}\!\!=\!\!100$ m, $300$ m, $400$ m. }\label{fig:comp_scm2}
\end{figure}

\begin{figure}	
	\begin{subfigure} {0.5\textwidth}
		\centering
		\includegraphics[ width=5.6cm]{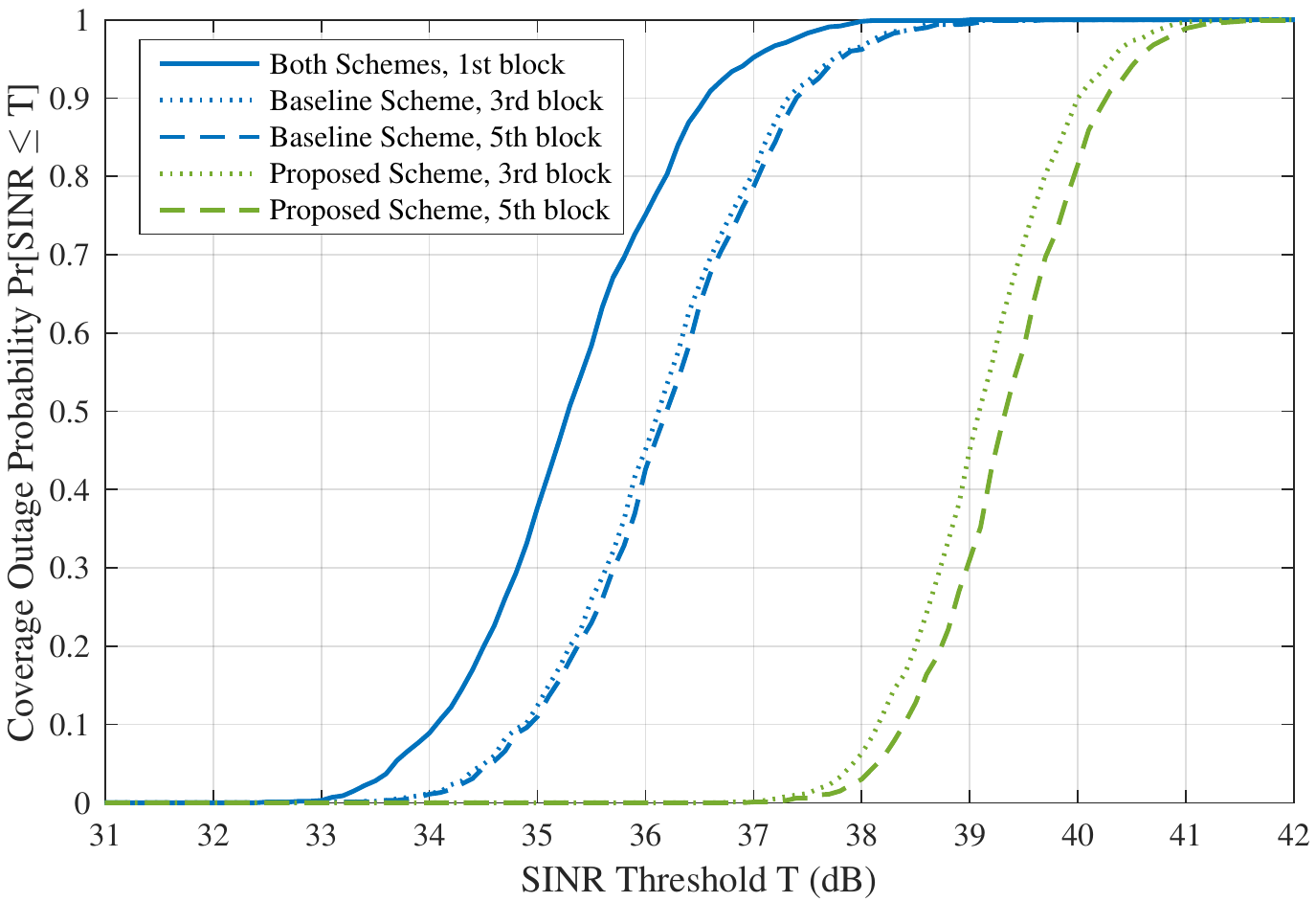}
		\centering \caption{$ r_{1,1}=100 $m}\label{fig:comp_center}
	\end{subfigure}
	\begin{subfigure} {0.5\textwidth}
		\centering
		\includegraphics[width=5.6cm]{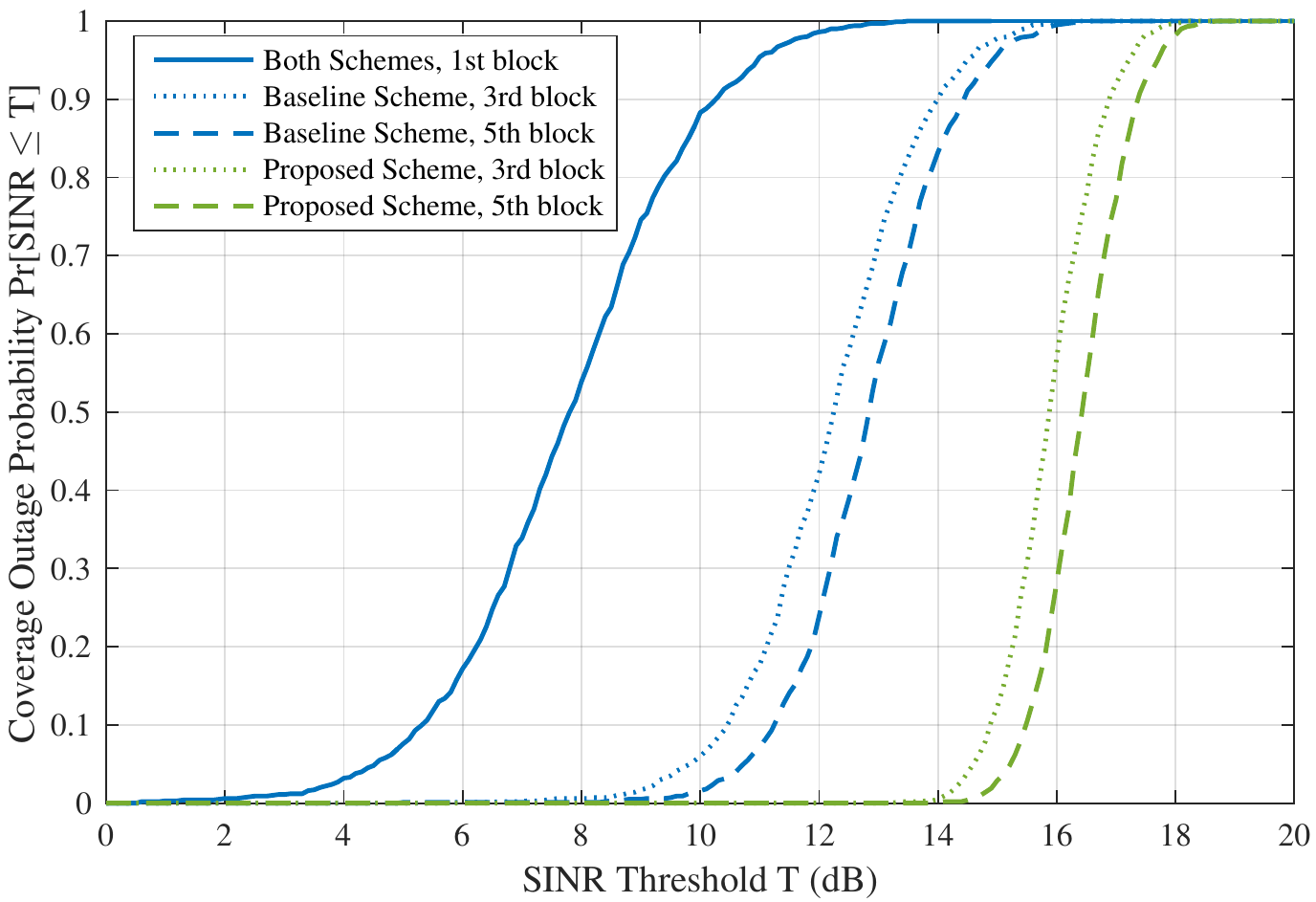}
		\centering \caption{$ r_{1,1}=400$m}\label{fig:comp_edge}
	\end{subfigure}
	\caption{SINR distributions of the target user with the proposed schemes, where the distributions of the $1$-st, $3$-rd and 5th blocks are plotted, $ M=200 $, and $ r_{1,1}= 100$ m, $400$ m.}\label{fig:comp_cdf}
\end{figure}

The performance of the two data-assisted schemes (proposed scheme and baseline scheme) is compared in Fig. \ref{fig:comp_cdf}(a) and \ref{fig:comp_cdf}(b) for cell center region $r_{1,1}=100$ m and cell edge region $ r_{1,1}=400$ m, respectively, where the backhaul delay $d=1$. The SINR distributions of the $1$-st, $3$-rd, $5$-th data blocks are plotted, respectively. The average number of users per cell is $ 10 $. The $ 1 $-st block's SINR distribution of the both schemes demonstrates the performance of massive MIMO system without data-assisted detection. Taking the $ 3 $-rd block of the baseline scheme as an example, when data symbols are used as equivalent pilot sequence, there are around $ 1 $dB and $ 5 $dB gains for the cell-center and cell-edge users, respectively. Hence, the data-assisted detection scheme benefits more on cell-edge users. This is because the cell-edge users suffer more on the inter-cell interference. In both figures, our proposed scheme outperforms the baseline scheme. This is because the pilot information and data symbols of the closest interfering users are used for channel estimation in the proposed scheme.  

%\begin{figure}	
%	\begin{subfigure} {0.5\textwidth}
%		\centering
%		\includegraphics[width=5.6cm]{blockidx_U10.pdf}
%		\centering \caption{The average number of users per cell is 10}\label{fig:block_U10}
%	\end{subfigure}
%	\begin{subfigure} {0.5\textwidth}
%		\centering
%	    \includegraphics[width=5.6cm]{blockidx_U20.pdf}
%		\centering \caption{The average number of users per cell is 20}\label{fig: block_U20}
%	\end{subfigure}
%	\caption{Uplink SINR of different blocks, where the SINR is chosen with coverage outage probability = $ 0.1\% $, $ M=200 $, $ r_{1,1}=100$m, $400$m.}\label{fig:block}
%\end{figure} 

\begin{figure}
	\centering
	\includegraphics[height = 4.5cm, width=6.6cm]{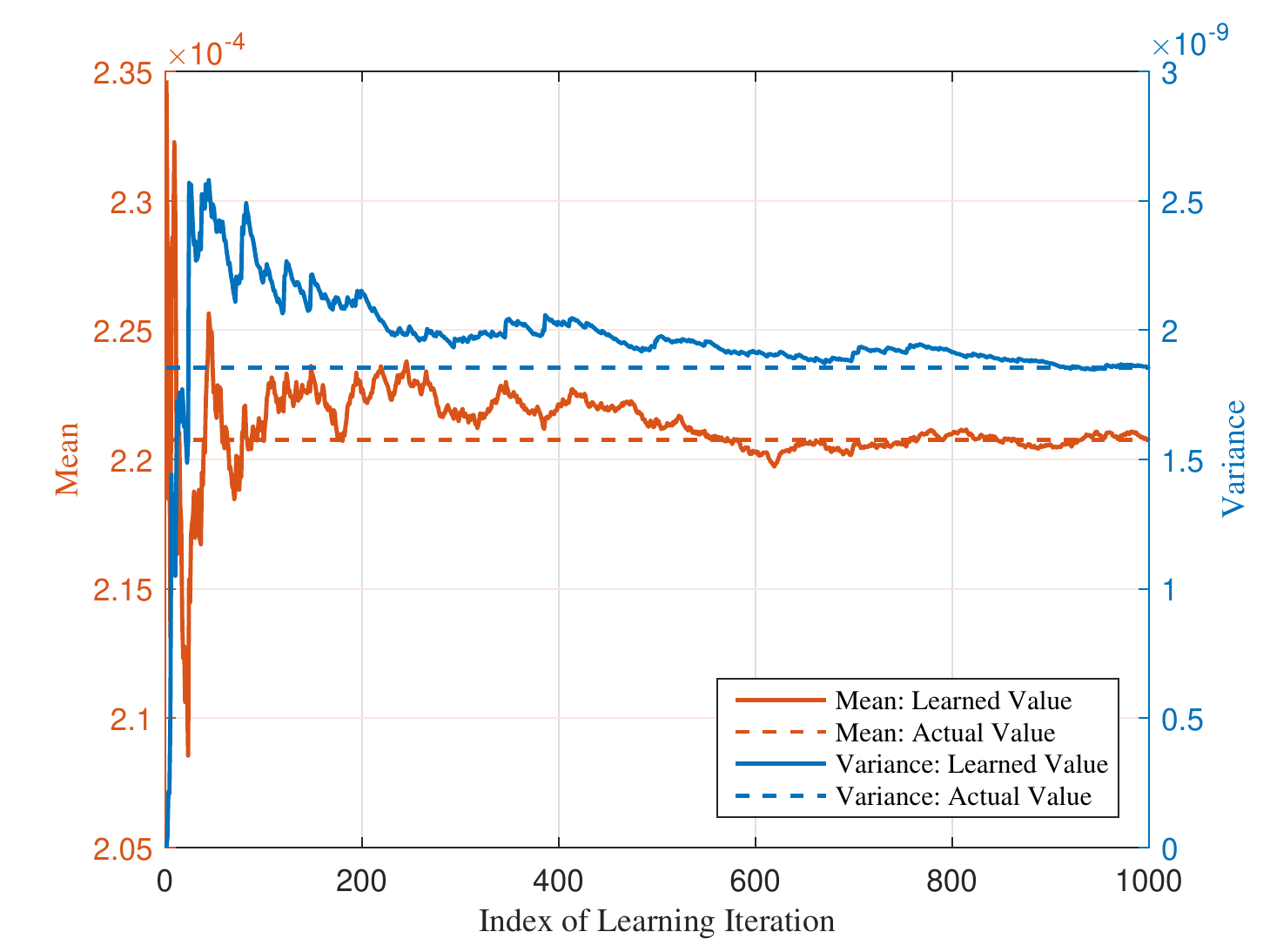}
	\caption{Online learning procedure of Theorem \ref{theo:sc1_1} at the $5$-th block,  where $M=200$, $ r_{1,1}=50$ m.}\label{fig:learn}
\end{figure}

%We further show the uplink SINR of each data block in Fig. \ref{fig:block}, where  the backhaul delay $d=1$, the SINR is chosen with coverage outage probability = $ 0.1\% $, and the average number of users per cell is $10$ and $20$, respectively. It is shown that regardless of users density, the proposed data-assisted approaches can effectively improve the ulplink SINR. Both the proposed scheme and baseline scheme could bring nearly $5$dB, $10$dB gains for the $3$-rd data block of cell-edge user respectively, compared with the conventional uplink transmission scheme without data-assisted channel estimation. Moreover, for the $4$-th and $5$-th data blocks, their SINR gains compared with the $3$-rd data block is negligible. Thus it is not necessary to involve a significant number of data symbols in the data-assisted channel estimation scheme.

In Fig. \ref{fig:learn}, we check the convergence of online learning algorithm, where the learning process of mean and variance in \eqref{eqn:cdf_s1} is demonstrated. The target user is $ 50 $ m away from the service BS. The average number of users per cell is $ 10 $. We can observe that after nearly two hundred times learning iterations, the mean and variance of interference power are close to the actual values.

\section{Summary}\label{sec:sum}
In this paper, a cooperative uplink transmission and detection scheme for massive MIMO networks is proposed and analyzed. Based on it, learning-based algorithms are proposed to adapt the uplink data rate without the knowledge of active user distribution density. It is shown by simulations that the proposed scheme can significantly improve the uplink SINRs of the back data blocks in each uplink frame. 

\section*{Appendix A: Proof of Lemma \ref{lem:sce1}}
We first consider the case of $i \leq d-1$.
Note that both $ \Delta \mathbf{R}_{1,k}^i $ and $ \mathbf{R}_{\ell,j} $ are scaled identity matrices, we could define $$ \alpha \mathbf{I} = \sum_{\forall k} \Delta \mathbf{R}_{1,k}^i  + \sum_{\forall \ell\neq 1, j} \mathbf{R}_{\ell,j} + \textbf{I}/P.$$ Moreover, we define the following notations for convenience:
\begin{align}
\mathbf{A}_k =& [ \sum_{\forall j \neq k} \widehat{\mathbf{h}}^i_{1,j} (\widehat{\mathbf{h}}^i_{1,j})^H + \alpha \mathbf{I} ]^{-1}, \nonumber\\
\mathbf{B}_{k,j} =& [ \sum_{\forall m \neq \{k, j\}} \widehat{\mathbf{h}}^i_{1,m} (\widehat{\mathbf{h}}^i_{1,m})^H + \alpha \mathbf{I} ]^{-1}. \nonumber
\end{align}
With the matrix inversion lemma, $ \mathbf{\Psi}^i_{in} $ can be written as
\begin{equation}
\mathbf{\Psi}^i_{in} = \alpha^{-1} \mathbf{I} - \alpha^{-2} \widehat{\mathbf{H}}_{1}^{i} \bigg( \mathbf{I} + \alpha^{-1} (\widehat{\mathbf{H}}_{1}^{i})^H \widehat{\mathbf{H}}_{1}^{i} \bigg)^{-1} (\widehat{\mathbf{H}}_{1}^{i})^H \label{eqn:phi}.
\end{equation}
Hence, we have the following asymptotic expressions for (\ref{eqn:ul-sinr}).

\textbullet \ The signal term:
\begin{align}
&|(\widehat{\mathbf{h}}_{1,1}^i)^H \mathbf{\Psi}^{i}_{in} \widehat{\mathbf{h}}_{1,1}^i|^2\nonumber\\
=& \bigg| \alpha^{-1} (\widehat{\mathbf{h}}_{1,1}^i)^H \widehat{\mathbf{h}}_{1,1}^i- \alpha^{-2} (\widehat{\mathbf{h}}_{1,1}^i)^H \widehat{\mathbf{H}}_{1}^{i} \nonumber\\
&\times \bigg( \mathbf{I} + \alpha^{-1} (\widehat{\mathbf{H}}_{1}^{i})^H \widehat{\mathbf{H}}_{1}^{i} \bigg)^{-1} (\widehat{\mathbf{H}}_{1}^{i})^H   \widehat{\mathbf{h}}_{1,1}^i \bigg|^2 \nonumber\\
\rightarrow& \bigg| \alpha^{-1} (\widehat{\mathbf{h}}_{1,1}^i)^H \widehat{\mathbf{h}}_{1,1}^i - \alpha^{-2} \frac{|(\widehat{\mathbf{h}}_{1,1}^i)^H \widehat{\mathbf{h}}_{1,1}^i|^2}{1 + \alpha^{-1} (\widehat{\mathbf{h}}_{1,1}^i)^H \widehat{\mathbf{h}}_{1,1}^i} \bigg|^2 
\rightarrow 1 \nonumber,
\end{align}
where the first step is due to (\ref{eqn:phi}); the second step is because $ |(\widehat{\mathbf{h}}_{1,1}^i)^H \widehat{\mathbf{h}}_{1,1}^i| >> |(\widehat{\mathbf{h}}_{1,1}^i)^H \widehat{\mathbf{h}}_{1,j}^i| $ ($ \forall j \neq 1 $) for sufficiently large $ M $ and $ L_i $, and $ \bigg( \mathbf{I} + \alpha^{-1} (\widehat{\mathbf{H}}_{1}^{i})^H \widehat{\mathbf{H}}_{1}^{i} \bigg) $ can be approximated as a diagonal matrix; the last step is because of $ \alpha^{-1} (\widehat{\mathbf{h}}_{1,1}^i)^H \widehat{\mathbf{h}}_{1,1}^i >>1 $ for sufficiently large $ M $.

\textbullet \ Interference term I:
\begin{eqnarray}
&&\sum_{\forall j \neq 1}|(\widehat{\mathbf{h}}_{1,1}^i)^H \mathbf{\Psi}^{i}_{in} \widehat{\mathbf{h}}_{1,j}^i|^2 \nonumber\\
& = & \sum_{\forall j \neq 1} \bigg| \frac{(\widehat{\mathbf{h}}_{1,1}^i)^H \mathbf{B}_{1,j} \widehat{\mathbf{h}}_{1,j}^i}{ [1 + (\widehat{\mathbf{h}}_{1,1}^i)^H \mathbf{A}_{1} \widehat{\mathbf{h}}_{1,1}^i ] [1 + (\widehat{\mathbf{h}}_{1,j}^i)^H \mathbf{B}_{1,j} \widehat{\mathbf{h}}_{1,j}^i ] }\bigg|^2 \nonumber \\
& \approx & \sum_{\forall j \neq 1} \bigg| \frac{(\widehat{\mathbf{h}}_{1,1}^i)^H \mathbf{B}_{1,j} \widehat{\mathbf{h}}_{1,j}^i}{ \alpha^{-2} ||\widehat{\mathbf{h}}_{1,1}||^2 ||\widehat{\mathbf{h}}_{1,j}||^2}\bigg|^2 \sim \mathcal{O}(\frac{1}{M^2}), \nonumber
\end{eqnarray}
where the first step is directly from the Lemma 1 of \cite{Hoydis:13}, and the last step is because the numerator and the denominator of $ \bigg| \frac{(\widehat{\mathbf{h}}_{1,1}^i)^H \mathbf{B}_{1,j} \widehat{\mathbf{h}}_{1,j}^i}{ \alpha^{-2} ||\widehat{\mathbf{h}}_{1,1}||^2 ||\widehat{\mathbf{h}}_{1,j}||^2}\bigg| $ are of the order $ M $ and $ M^2 $ respectively. To prove the second step of the above derivation, we define $ \widehat{\mathbf{C}}_{1,1}^i = [\widehat{\mathbf{h}}_{1,2}^i \widehat{\mathbf{h}}_{1,3}^i \cdots \widehat{\mathbf{h}}_{1,|\Phi^{(1)}|}^i] $. Hence,
\begin{align}
1 + (\widehat{\mathbf{h}}_{1,1}^i)^H \mathbf{A}_{1} \widehat{\mathbf{h}}_{1,1}^i 
=&  1+ \alpha^{-1} (\widehat{\mathbf{h}}_{1,1}^i)^H \widehat{\mathbf{h}}_{1,1}^i - \alpha^{-2} (\widehat{\mathbf{h}}_{1,1}^i )^H \nonumber\\
&\widehat{\mathbf{C}}_{1,1}^i\! \bigg[\! \mathbf{I}\! +\! \alpha^{\!-1}\!(\widehat{\mathbf{C}}_{1,1}^i)^H  \widehat{\mathbf{C}}_{1,1}^i  \!\bigg]^{\!-1}\!\!\!\! (\widehat{\mathbf{C}}_{1,1}^i)^{\!H}  \widehat{\mathbf{h}}_{1,1}^i \nonumber\\
\approx& \alpha^{-1} (\widehat{\mathbf{h}}_{1,1}^i)^H \widehat{\mathbf{h}}_{1,1}^i \nonumber,
\end{align}
and similarly $$
1 + (\widehat{\mathbf{h}}_{1,j}^i)^H \mathbf{B}_{1,j} \widehat{\mathbf{h}}_{1,j}^i \approx \alpha^{-1} (\widehat{\mathbf{h}}_{1,j}^i)^H \widehat{\mathbf{h}}_{1,j}^i.$$

\textbullet \ Interference term III:

The intra-cell interference can be written as
\begin{eqnarray}
&&\mathbb{E} \bigg[\sum_{(\ell,j) \notin \Phi^{1}} |(\widehat{\mathbf{h}}_{1,1}^i)^H \mathbf{\Psi}^{i}_{in} \mathbf{h}_{\ell,j}|^2  \bigg| \widehat{\mathbf{H}}_{1}^i\bigg]  \nonumber\\
&\overset{(a)}{=}& \mathbb{E} \bigg[\sum_{(\ell,j) \notin \Phi^{1}} \bigg| \alpha^{-1}(\widehat{\mathbf{h}}_{1,1}^i)^H \mathbf{h}_{\ell,j}    - \alpha^{-2} (\widehat{\mathbf{h}}_{1,1}^i)^H \widehat{\mathbf{H}}_{1}^{i}\nonumber\\
&&\times \bigg( \mathbf{I} + \alpha^{-1} (\widehat{\mathbf{H}}_{1}^{i})^H \widehat{\mathbf{H}}_{1}^{i} \bigg)^{-1} (\widehat{\mathbf{H}}_{1}^{i})^H  \mathbf{h}_{\ell,j} \bigg|^2 \bigg| \widehat{\mathbf{H}}_{1}^i\bigg] \nonumber\\
& \overset{(b)}{\approx} &  \mathbb{E} \bigg[\sum_{(\ell,j) \notin \Phi^{1}} \bigg| \alpha^{-1}(\widehat{\mathbf{h}}_{1,1}^i)^H \mathbf{h}_{\ell,j}  - \alpha^{-2} (\widehat{\mathbf{h}}_{1,1}^i)^H \widehat{\mathbf{h}}_{1,1}^{i}\nonumber\\
&&\times \bigg( 1 + \alpha^{-1} (\widehat{\mathbf{h}}_{1,1}^{i})^H \widehat{\mathbf{h}}_{1,1}^{i} \bigg)^{-1} (\widehat{\mathbf{h}}_{1,1}^{i})^H  \mathbf{h}_{\ell,j} \bigg|^2 \bigg| \widehat{\mathbf{H}}_{1}^i\bigg] \nonumber\\
&\overset{(c)}{\rightarrow}&   \frac{1}{ \rho_{1,1}^2}\mathbb{E} \bigg[ \sum_{(\ell,j) \notin \Phi^{1}} \bigg|  (\widehat{\mathbf{h}}_{1,1}^i)^H  \mathbf{h}_{\ell,j} / M\bigg|^2 \ \bigg| \widehat{\mathbf{H}}_{1}^i\bigg] \nonumber\\
&\approx& \frac{1}{ \rho_{1,1}^2} \sum_{(\ell,j) \notin \Phi^{1}} \bigg\{ \mathbb{E}\bigg|\mathbf{h}_{1,1}^H \mathbf{h}_{\ell,j}/M \bigg|^2 + \bigg[ ||\mathbf{q}_{1}^i||^2P\rho_{\ell,j}^2 \bigg] \bigg\}\nonumber\\
&\overset{(d)}{\approx}& \frac{1}{\rho_{1,1}^2} \!\!\sum_{(\ell,j) \notin \Phi^{1}}\!\! \bigg\{ \frac{\rho_{1,1} \rho_{\ell,j}}{M} \!+\! P\rho_{\ell,j}^2 \bigg[ \frac{1}{||\mathbf{x}_{1,1}^{0,i}||^2} \bigg] \bigg\} \nonumber\\
&\overset{(e)}\approx&   \frac{1}{\rho_{1,1}^2} \sum_{(\ell,j) \notin \Phi^{1}} \!\!\bigg\{ \frac{\rho_{1,1} \rho_{\ell,j}}{M} \!+\! \frac{\rho_{\ell,j}^2}{(L_i -1)}  \bigg\}, \label{equ:inter_term}
\end{eqnarray}
where
\begin{eqnarray}
||\mathbf{q}_k^i||^2
&\approx&\bigg[  \frac{||\mathbf{x}_{1,k}^{0,i}||}{\frac{P\sum_{\ell\neq 1} \rho_{\ell,j}}{\rho_{1,1}}+ ||\mathbf{x}_{1,k}^{0,i}||^2} \bigg]^2 \approx \frac{1}{||\mathbf{x}_{1,k}^{0,i}||^2}. \label{equ:estimator_term}
\end{eqnarray}
Step (a) is due to (\ref{eqn:phi}), step (b) is because of $ |(\widehat{\mathbf{h}}_{1,1}^i)^H \widehat{\mathbf{h}}_{1,1}^i| >> |(\widehat{\mathbf{h}}_{1,1}^i)^H \widehat{\mathbf{h}}_{1,j}^i| $ ($ \forall j \neq 1 $) for sufficiently large BS's antennas $ M $, step (c) is due to $ \alpha^{-1} (\widehat{\mathbf{h}}_{1,1}^i)^H \widehat{\mathbf{h}}_{1,1}^i >> 1$ and $ (\widehat{\mathbf{h}}_{1,1}^i)^H \widehat{\mathbf{h}}_{1,1}^i/M \rightarrow \rho_{1,1} $ for sufficiently large $ L_i $ and $ M $. For sufficiently large $L_i$, the channel estimation error in $\widehat{\mathbf{h}}_{1,1}^i$ is much smaller than its true value. The step (d) is due to equation (\ref{equ:estimator_term}), and step (e) is because of the mean of inverse-Chi-squared distribution.

\setcounter{equation}{31}
\textbullet \ Interference term II:
\begin{eqnarray}
&&\mathbb{E}\bigg[\sum\limits_{\forall k} | {(\widehat {\bf{h}}_{1,1}^i)^H}{\bf{\Psi }}_{in}^i\Delta {\bf{h}}_{1,k}^i{|^2}|\widehat {\bf{H}}_1^i\bigg] \nonumber\\ 
&=&\mathbb{E}\bigg[\sum\limits_{\forall k} | {(\widehat {\bf{h}}_{1,1}^i)^H}{\bf{\Psi }}_{in}^i{{\bf{h}}_{1,k}}\left( {{\bf{x}}_{1,k}^{0,i}{\bf{q}}_k^i - 1} \right)\nonumber\\
&&+\sum\limits_{(\ell,j) \notin {\Phi^1}} {{{(\widehat {\bf{h}}_{1,1}^i)}^H}{{\bf{h}}_{\ell,j}}{\bf{x}}_{\ell,j}^{0,i}{\bf{q}}_k^i} {|^2}|\widehat {\bf{H}}_1^i\bigg] \nonumber\\
&\approx &\mathbb{E}\bigg[\sum\limits_{\forall k} {P{{\left\| {{\bf{q}}_k^i} \right\|}^2}\sum\limits_{(\ell,j) \notin {\Phi ^1}} {|{{(\widehat {\bf{h}}_{1,1}^i)}^H}{\bf{\Psi }}_{in}^i{{\bf{h}}_{\ell,j}}} {|^2}} |\widehat {\bf{H}}_1^i\bigg], \nonumber
\end{eqnarray}
where the approximation is due to ${\bf{x}}_{1,k}^{0,i}{\bf{q}}_k^i \approx 1$. From (\ref{equ:inter_term}) and (\ref{equ:estimator_term}), we have
\begin{eqnarray}
&&\mathbb{E}\bigg[\sum\limits_{\forall k} | {(\widehat {\bf{h}}_{1,1}^i)^H}{\bf{\Psi }}_{in}^i\Delta {\bf{h}}_{1,k}^i{|^2}|\widehat {\bf{H}}_1^i\bigg]\nonumber\\
&\approx&  \!\!\sum\limits_{\forall k}\frac{P}{{{{\left\| {{\bf{x}}_{1,k}^{0,i}} \right\|}^2}}}\sum\limits_{(\ell,j) \notin {\Phi^1}}  [ \frac{{{\rho _{\ell,j}}}}{{\rho _{1,1}M}} + \frac{{\rho _{\ell,j}^2}}{{\rho _{1,1}^2({L_i} - 1)}}] \nonumber\\
&=&  \frac{|\Phi^1|}{{({L_i} - 1)}}\sum\limits_{(\ell,j) \notin {\Phi^1}} [ \frac{{{\rho _{\ell,j}}}}{{\rho _{1,1}M}} + \frac{{\rho _{\ell,j}^2}}{{\rho _{1,1}^2({L_i} - 1)}} ]. \label{equ:intra_term}
\end{eqnarray}

\textbullet \ Noise term: $
\frac{\left\|(\widehat{\mathbf{h}}_{1,1}^{i})^H \mathbf{\Psi}^{i}_{in}\right\|^2}{P} \approx 
\frac{1}{P}{\left\| {\frac{{{\alpha ^{ - 1}}{{(\widehat {\bf{h}}_{1,1}^i)}^H}}}{{1 + {\alpha ^{ - 1}}{{(\widehat {\bf{h}}_{1,1}^i)}^H}\widehat {\bf{h}}_{1,1}^i}}} \right\|^2}.$
When $L_i$ is sufficiently large, the noise item can be approximated by $\frac{1}{{PM{\rho _{1,1}}}}$.

Note that for high SNR region and sufficiently large $ M $, the noise item and interference term I are negligible compared with the interference term II and III. As $ L_i \approx  L_i -1 $, the asymptotic expression of  $ \gamma_{1,1}^{i+1} $ can be derived from (\ref{equ:inter_term}) and (\ref{equ:intra_term}). 
When $i \ge d $, the proof is similar, due to page limitation, it is omitted here.

\section*{Appendix B: Proof of Theorem \ref{theo:sc1_1}}

Due to page limitation, we only provide the sketch of the proof for $i \leq d-1$, it can be similarly applied on the case of $i \ge d$. The region of the interfering users $\bar S_1$ can be viewed as a disk, whose radius $ r $ tends to infinity, except the coverage of target cell. Hence in the following, we first consider $\bar S_1$ as a disk with fixed radius, and then let $ r\rightarrow \infty $. Denote the set of users in $ \bar S_1 $ as $\Phi_{intf}$. Let $ \alpha \leq 1 $ and $ \beta \geq 1 $, we have
\begin{eqnarray}
&&{P_{\alpha , < }}\mathop  = \limits^{def} \Pr [|{\Phi _{intf}}| < \alpha \int\limits_{{{\bar S}_1}} {{\lambda _u}({\bf{l}})}ds{({\bf l})} ] \nonumber\\
=&& \frac{{\Gamma (\lfloor \alpha \int\limits_{{{\bar S}_1}} {{\lambda _u}({\bf{l}})ds{({\bf l})}}  + 1 \rfloor,\int\limits_{{{\bar S}_1}} {{\lambda _u}({\bf{l}})}ds{({\bf l})} )}}{{\lfloor \alpha \int\limits_{{{\bar S}_1}} {{\lambda _u}({\bf{l}})ds{({\bf l})}}\rfloor !}},\nonumber
\end{eqnarray}
\begin{eqnarray}
&&P_{\beta,>} \overset{def}{=} \Pr[|\Phi_{intf}| > \beta \int\limits_{{{\bar S}_1}} {{\lambda _u}({\bf{l}})ds{({\bf l})}} ]\nonumber\\
=&&1-\frac{\Gamma \bigg(\lfloor  \beta \int\limits_{{{\bar S}_1}} {{\lambda _u}({\bf{l}})ds{({\bf l})}}+1 \rfloor, \int\limits_{{{\bar S}_1}} {{\lambda _u}({\bf{l}})ds{({\bf l})}}\bigg)}{\lfloor \beta \int\limits_{{{\bar S}_1}} {{\lambda _u}({\bf{l}})ds{({\bf l})}} \rfloor!} ,\nonumber
\end{eqnarray}
where $ \Gamma(.) $ is the incomplete gamma function. Let $ \mathcal{E}_{intf} = \sum\limits_{(j,m) \in \Phi_{intf}}  \frac{1}{M} \frac{\rho_{j,m}}{\rho_{1,1} } + \frac{1}{L_i} (\frac{\rho_{j,m}}{\rho_{1,1}})^2$ and $T^{\prime}$ = ${T (\frac{{ {\Phi^\dagger_i} }}{{{L_i}}} + 1) }$, it is with probability $ 1-P_{\alpha,<} $ and $ 1- P_{\beta,>}$ respectively that the following equations can hold 
\begin{equation}
\Pr\bigg[ \mathcal{E}_{intf} > \frac{1}{T^{\prime}} \bigg| |\Phi_{intf}|=\alpha \int\limits_{{{\bar S}_1}} {{\lambda _u}({\bf{l}})ds{({\bf l})}} \bigg] \leq \Pr\bigg[ \mathcal{E}_{intf} > \frac{1}{T^{'}} \bigg] \nonumber,
\end{equation}
\begin{equation}
\Pr\bigg[\mathcal{E}_{intf} > \frac{1}{T^{\prime}}\bigg]  \leq  \Pr\bigg[ \mathcal{E}_{intf} > \frac{1}{T^{\prime}}\bigg| |\Phi_{intf}|=\beta \int\limits_{{{\bar S}_1}} {{\lambda _u}({\bf{l}})ds{({\bf l})}} \bigg] \nonumber.
\end{equation}
%\begin{figure}
%	\centering
%	\includegraphics[scale=0.28]{Proof_Theo1.pdf}
%	\caption{Illustration of the proof of Theorem 1.}\label{fig:proof1}
%\end{figure}

When $\rho_{1,1}$ is given, $ \frac{1}{M} \frac{\rho_{j,m}}{\rho_{1,1} } + \frac{1}{L_i} (\frac{\rho_{j,m}}{\rho_{1,1} })^2, \forall (j,m) \in \Phi_{intf} $ could be consider as i.i.d. distributed random variables. Given the condition $ |\Phi_{intf}|=\alpha \int\limits_{{{\bar S}_1}} {{\lambda _u}({\bf{l}})ds{({\bf l})}}$, the above lower bound can be further approximated by applying the Central Limit Theorem as 
\begin{align}
&\Pr\bigg[ \mathcal{E}_{intf} > \frac{1}{T^{\prime}} \bigg| |\Phi_{intf}|=\alpha \int\limits_{{{\bar S}_1}} {{\lambda _u}({\bf{l}})ds{({\bf l})}} \bigg] \nonumber\\
=& Q \bigg[ \bigg(\frac{1}{T^{\prime}} - \alpha \mathcal{M}_{i+1} \bigg) \sqrt{\frac{1}{\alpha \mathcal{V}_{i+1}}} \bigg]\nonumber.
\end{align}
where $\mathcal{M}_{i+1}$ and $\mathcal{V}_{i+1}$ are derived by the following integral expressions
\begin{equation}
\mathcal{M}_{i+1} =  \int\limits_{{\bar S_1}} \left[{\frac{e^{\frac{a^2\theta^2}{2}}}{M}\frac{{{|{\bf l}|^{-\sigma}}}}{{\rho}_{1,1}} + \frac{e^{2a^2\theta^2}}{{L_i}}{{(\frac{{{|{\bf l}|^{-\sigma} }}}{{\rho}_{1,1} })}^2}}\right] {{\lambda _u}({\bf{l}})}ds{({\bf l})}, \nonumber
\end{equation}
\begin{equation}
\mathcal{V}_{i+1} =\!\! \int\limits_{{\bar S_1}}\!\! 
\left[\frac{e^{\frac{a^2\theta^2}{2}}}{M}\frac{{{|{\bf l}|^{-\sigma}}}}{{\rho}_{1,1}}\! +\! \frac{e^{2a^2\theta^2 }}{{L_i}}{{(\frac{{{|{\bf l}|^{-\sigma} }}}{{\rho}_{1,1} })}^2}\right]^2\!\!\!{{\lambda _u}({\bf{l}})}ds{({\bf l})} - \mathcal{M}_{i+1}^2. \nonumber
\end{equation}
When $r \rightarrow \infty $, however, $\mathcal{M}_{i+1}$ and $\mathcal{V}_{i+1}$ will converge to finite values.
Moreover, $ P_{a,<} \rightarrow 1 $ when $r \rightarrow \infty $.
Thus, it is with probability $ 1 $ that for arbitrary $ \alpha <1 $, $$
\Pr\bigg[ \mathcal{E}_{intf} > \frac{1}{T^{\prime}} \bigg]
\geq  Q \bigg[ \bigg(\frac{1}{T^{\prime}} - \alpha  \mathcal{M}_{i+1} \bigg) \sqrt{\frac{1}{\alpha  \mathcal{V}_{i+1}}} \bigg]. $$
Similarly, it is with probability $ 1 $ that for arbitrary $ \beta > 1 $, $$
\Pr\bigg[ \mathcal{E}_{intf} > \frac{1}{T^{\prime}} \bigg]
\leq Q \bigg[ \bigg(\frac{1}{T^{\prime}} - \beta  \mathcal{M}_{i+1} \bigg) \sqrt{\frac{1}{\beta  \mathcal{V}_{i+1}}} \bigg]. $$ Thus, it is with probability $ 1 $ that $$
\Pr\bigg[ \mathcal{E}_{intf} > \frac{1}{T^{\prime}} \bigg]
= Q \bigg[ \bigg(\frac{1}{T^{\prime}} -  \mathcal{M}_{i+1} \bigg) \sqrt{\frac{1}{ \mathcal{V}_{i+1}}} \bigg]. $$
Finally, we have $
\Pr \bigg[ \gamma_{1,1}^{i+1} < T \bigg]  = \Pr\bigg[ \mathcal{E}_{intf} > \frac{1}{T(\frac{{\left| {\Phi^\dagger_i} \right|}}{{{L_i}}} + 1)} \bigg] = Q \bigg[ \bigg(\frac{1}{T (\frac{{ {\Phi^\dagger_i} }}{{{L_i}}} + 1) } - \mathcal{M}_{i+1} \bigg) \sqrt{\frac{1}{\mathcal{V}_{i+1}}} \bigg]. $

\section*{Appendix C: Proof of Lemma \ref{Ext_M&V}}
Without loss of generality, we only provide the proof for $ b=1 $. Since large-scale fadings $\rho_{1,m}$ and $\rho_{1,j}$ are known to the target BS, it can be deduced from \eqref{DefM} and \eqref{DefV} that	 
\begin{eqnarray}
	\sum\limits_{\forall (\ell,n) \notin \Phi^\dagger_{i-1}}\!\!\!\!\!\!\mathbb{E} \left[ \rho_{\ell,n}  \right] \!\!\!& = &\!\!\! \frac{{M({\cal M}_i^2(1,i)\rho _{1,m}^2 - {\cal M}_i^2(1,j)\rho _{1,j}^2)}}{{{\rho _{1,m}} - {\rho _{1,j}}}},\nonumber \\
	\sum\limits_{\forall (\ell,n) \notin \Phi^\dagger_{i-1}}\!\!\!\!\!\!\mathbb{E} \left[ \rho_{\ell,n}^2 \right] \!\!\!\!& = &\!\!\!\! \frac{{{\!L_i^\dagger}{\rho _{1,m}}{\rho _{1,j}}(\!{{\cal M}_i}(\!1,m\!){\rho _{1,m}} \!\!-\!\! {{\cal M}_i}(\!1,j\!){\rho _{1,j}}\!)}}{{{\rho _{1,j}} \!-\! {\rho _{1,m}}}}. \nonumber			
\end{eqnarray}
Note that for arbitrary $(1,k)$-th user, the mean $\mathcal{M}_i(1,k) = \mathbb{E}\!\!\! \sum\limits_{\forall (\ell,n) \notin \Phi^\dagger_{i-1}}\!\!\!  \left[  \frac{1}{M} \frac{\rho_{\ell,n}}{\rho_{1,k} } \!+ \!\frac{1}{L_i^\dagger} (\frac{\rho_{\ell,n}}{\rho_{1,k} })^2\right]$. Substituting the above two equations into the definition of ${\cal M}_i(1,k)$, \eqref{DeduceM} can be proved.
Moreover, the variance $\mathcal{V}_i(1,k)$ can be obtained similarly.

%-----------------------------------------------------------
\bibliographystyle{IEEEtran}
\bibliography{ray}
%-----------------------------------------------------------

\end{document}